\documentclass{LMCS}

\def\doi{9(2:09)2013}
\lmcsheading%
{\doi}
{1--19}
{}
{}
{Nov.~30, 2012}
{Jun.~25, 2013}
{}

\usepackage{enumerate,hyperref}
\usepackage{pstricks,color,amsthm}
\usepackage{algpseudocode,fixltx2e}
\usepackage{url}
\MakeRobust{\Call}

\newcommand{\proc}[1]{{\textsc{#1}}}
\newcommand{\f}{\ensuremath{\varphi}}
\newcommand{\p}{\ensuremath{\psi}}
\renewcommand{\th}{\ensuremath{\theta}}
\newcommand{\si}{\sigma}
\newcommand{\De}{\Delta}
\newcommand{\Ga}{\Gamma}
\newcommand{\Si}{\Sigma}
\newcommand{\lang}{{\mathcal L}}
\newcommand{\vty}[1]{\mathsf{#1}}
\newcommand{\op}[1]{\mathbb{#1}}
\newcommand{\lgc}[1]{\mathrm{#1}}
\newcommand{\alg}[1]{\mathbf{#1}}
\newcommand{\mdl}[1]{\models_{\lgc{#1}}}
\newcommand{\der}[1]{\vdash_{\lgc{#1}}}
\newcommand{\Con}{{\rm Con}}
\newcommand{\imp}{\Rightarrow}
\newcommand{\rl}{\ / \ }

\newcommand{\eq}{\approx}
\newcommand{\Qg}{{\ensuremath{\mathbb{Q}}}}
\newcommand{\Vg}{{\ensuremath{\mathbb{V}}}}
\newcommand{\nVg}{{\ensuremath{\mathbb{V}^{\mbox{-}}}}}
\newcommand{\K}{{\ensuremath{\mathcal{K}}}}
\newcommand{\F}{{\ensuremath{\alg{F}}}}
\newcommand{\Q}{{\ensuremath{\mathcal{Q}}}}

\newcommand{\A}{{\ensuremath{\alg{A}}}}
\newcommand{\B}{{\ensuremath{\alg{B}}}}
\newcommand{\C}{{\ensuremath{\alg{C}}}}
\newcommand{\Luk}{\textup{\L}}

\begin{document}

\title[Admissibility in Finitely Generated Quasivarieties]{Admissibility in Finitely Generated Quasivarieties\rsuper*}

\author[G.~Metcalfe]{George Metcalfe}	
\address{Mathematics Institute, University of Bern, Sidlerstrasse 5, Bern 3012, Switzerland}	
\email{\{george.metcalfe, christoph.roethlisberger\}@math.unibe.ch}  
\thanks{The authors acknowledge support from Swiss National Science Foundation grants 20002{\_}129507 
and 200021{\_}146748 and 
Marie Curie Reintegration Grant PIRG06-GA-2009-256492.}	
\author[C.~R{\"o}thlisberger]{Christoph R{\"o}thlisberger}	

\ACMCCS{[{\bf Theory of computation}]: Models of computation; [{\bf Computing Methodologies}]:  Symbolic and algebraic manipulation---Symbolic and algebraic algorithms---Theorem proving algorithms}
\titlecomment{{\lsuper *}A precursor to this paper, reporting preliminary results, has appeared as~\cite{MR12}.}


\begin{abstract}
Checking the admissibility of quasiequations in a finitely generated
(i.e., generated by a finite set of finite algebras) quasivariety $\Q$
amounts to checking validity in a suitable finite free algebra of the
quasivariety, and is therefore decidable. However, since free algebras
may be large even for small sets of small algebras and very few
generators, this naive method for checking admissibility in $\Q$ is
not computationally feasible. In this paper, algorithms are introduced
that generate a minimal (with respect to a multiset well-ordering on
their cardinalities) finite set of algebras such that the validity of
a quasiequation in this set corresponds to admissibility of the
quasiequation in $\Q$. In particular, structural completeness
(validity and admissibility coincide) and almost structural
completeness (validity and admissibility coincide for quasiequations
with unifiable premises) can be checked. The algorithms are
illustrated with a selection of well-known finitely generated
quasivarieties, and adapted to handle also admissibility of rules in
finite-valued logics.
\end{abstract}


\maketitle

 
\section{Introduction}

The problem of checking the {\em validity} of quasiequations in finitely generated (i.e., generated by a finite set of finite algebras) 
quasivarieties or, similarly, checking consequences from finite sets of formulas 
in finite-valued logics, is decidable and has been  investigated extensively 
in the literature. In particular, uniform methods for generating proof systems to check validity such as tableaux, resolution, and multisequent calculi, 
have been developed, as have standard optimization techniques for these systems such as lemma generation and indexing (see, e.g.,~\cite{Hah93,Zac93,BFS99}).  
However, checking the {\em admissibility}  of quasiequations in finitely generated quasivarieties, or similarly, checking the 
admissibility of rules in finite-valued logics, 
is not so well-understood. The problem is decidable, but a naive approach leads to computationally 
unfeasible procedures even for small sets of small algebras. The main goal of this paper is to define uniform methods 
that generate computationally acceptable proof systems for checking admissibility in an arbitrary 
finitely generated quasivariety.

Intuitively, a rule is said to be admissible in a logical system if it can be added to the system without producing any new theorems. 
More formally, a quasiequation is admissible in a class of algebras $\K$ if every $\K$-unifier of the premises is a $\K$-unifier 
of the conclusion, where a $\K$-unifier of an equation $\f \eq \p$ is a substitution $\si$ such that $\si(\f) \eq \si(\p)$ is valid in 
$\K$. Admissibility plays a fundamental meta-level role in describing ``hidden properties'' of classes of  
algebras and logical systems. For example, establishing the completeness of a logical system  with respect to some restricted 
class of algebras (perhaps just one standard algebra) often involves showing that a certain rule or 
quasiequation is admissible; see, e.g.,~\cite{MOG08} for applications of the admissibility of rules 
in the context of fuzzy logics. Also, the closely related problem of deciding unifiability of concepts can be a 
useful tool for database redundancy checking in description logics~\cite{BM10}. 
Moreover, it may be possible to automatically obtain admissible rules 
for classes of algebras and logics that can then be used to simplify reasoning steps or to speed up 
derivations for checking validity.

Admissibility (in tandem with unification) has been studied intensively 
in the context of intermediate and transitive modal logics and their 
algebras~\cite{Ryb97,Iem01,Ghi99,Ghi00,Jer05,CM10,BR11a,BR11b}, 
leading in some cases to proof systems for checking admissibility~\cite{Ghi02,IM09a,IM09b,BRST10}. 
Axiomatizations and characterizations have also been obtained for certain families of 
finite algebras and many-valued logics, in particular \L ukasiewicz logics (or classes of 
MV-algebras)~\cite{Jer09a,Jer09b} 
and other fuzzy logics~\cite{CM09}, fragments of the substructural 
logic R-Mingle~\cite{Met12a}, and classes of De Morgan algebras~\cite{MR12a,CM201x}. However, a general theory, 
covering arbitrary finite algebras and finite-valued logics, has so far been lacking. 

The starting point for this work is the observation (see Lemma 4.1.10 of~\cite{Ryb97} and Corollary~\ref{c:admfin} below) 
that for a finite set of finite algebras $\K$, admissibility in the quasivariety $\Qg(\K)$ amounts to validity in the free algebra  on 
$n$ generators $\F_\K(n)$, where $n$ is the maximum cardinality of the algebras in $\K$. Since by Birkhoff's theorem on the structure of 
free algebras~\cite{Bir35}, this algebra $\F_\K(n)$ is finite, checking admissibility in $\Qg(\K)$ is decidable.
On the other hand, even for small $n$ and a small set of small algebras $\K$, 
the size of $\F_\K(n)$ may be prohibitively large for checking validity.
This is striking since validity and admissibility in $\Qg(\K)$ may coincide, 
$\Qg(\K)$ is then called {\em structurally complete}, or at least coincide for quasiequations 
with $\Qg(\K)$-unifiable premises, in which case, $\Qg(\K)$ is called {\em almost structurally complete}.
In other cases, $\Qg(\K)$-admissibility may correspond
to validity in  other, often quite small, algebras.  
We provide general algorithms here that discover such algebras, or, more precisely, 
generate finite sets of finite algebras such that the $\Qg(\K)$-admissibility of a quasiequation 
corresponds to validity in the quasivariety generated by these algebras. It is shown, moreover, 
that these are the smallest sets of algebras with this property with respect to a standard 
well-ordering on the multiset of their cardinalities.

We proceed as follows. First, in Section~\ref{s:prelim}, we recall some basic notions from universal algebra. 
Then  in Section~\ref{s:fingen}, we introduce some key ideas and methods for finitely generated quasivarieties; 
in particular, we apply a standard multiset well-ordering to the cardinalities of algebras in generating sets for 
quasivarieties, and provide an algorithm for finding the (unique up to isomorphism) minimal 
generating set of a finitely generated quasivariety. Section~\ref{s:admfree} provides characterizations of admissibility,
 unifiability, structural completeness, and almost structural completeness. These characterizations are then exploited 
 in Section~\ref{s:algorithms} to define corresponding algorithms, and illustrated  
using a selection of well-known finite algebras, confirming some known results 
from the literature, and establishing new ones. In Section~\ref{s:logics}, the approach is extended to 
finite-valued logics, where the designated values as well as the underlying finite algebra play a significant role. 
Finally, in Section~\ref{s:concluding}, we conclude with some remarks on future directions for this research.


 \section{Preliminaries} \label{s:prelim}

Let us first recall some basic ideas from universal algebra, referring to~\cite{BS81,Gor98} for further details. 
Given an algebraic language $\lang$ (i.e., without relation symbols), an {\em $\lang$-algebra} $\A$ is 
an algebraic structure consisting of a set $A$ (the universe) and an $n$-ary function $\star^\A$ corresponding to each 
$n$-ary function symbol $\star$ of $\lang$ (as usual, calling nullary functions  \emph{constants}).
We call $\A$ {\em finite} if $A$ is a finite set and $\lang$ consists of finitely many function symbols.  
A {\em congruence} on an $\lang$-algebra $\A$ is an equivalence relation $\th$ on $A$ satisfying for each $n$-ary 
function symbol $\star$ of $\lang$: $\{(a_1,b_1),\ldots,(a_n,b_n)\} \subseteq \th$ implies 
$(\star^\A(a_1,\ldots,a_n),\star^\A(b_1,\ldots,b_n)) \in \th$. 
The congruences of $\A$ form a complete lattice $\Con(\A)$ with bottom element $\De_\A = \{(a,a)  :  a \in A\}$ 
and top element $\nabla_\A = \{(a,b)  :  a,b \in A\}$, where the meet of a set of congruences on $\A$ is just the 
intersection of those congruences. Given $\th \in \Con(\A)$, the {\em quotient algebra of $\A$ by $\th$} 
is the $\lang$-algebra $\A / \th$ with universe $A / \th$ consisting of the equivalence classes $a / \th$ for $a \in A$ 
with functions defined for each $n$-ary function symbol $\star$ of $\lang$ by 
$\star^{\A / \th}(a_1 / \th,\ldots, a_n / \th) = \star^{\A}(a_1,\ldots,a_n) / \th$.

Term algebras $\alg{Tm_\lang}(X)$ are defined over a set of variables $X$ in the usual way, 
writing just $\alg{Tm_\lang}$ when $X$ is a fixed  countably infinite set, and letting $\f,\p$ stand for arbitrary members of the 
universe ${\rm Tm}_\lang$ called {\em $\lang$-terms}.
An {\em $\lang$-equation} is an ordered pair of $\lang$-terms, written $\f \eq \p$. 
An {\em $\lang$-clause} is defined as an ordered pair $\Si,\De$
of finite sets of $\lang$-equations, written $\Si \imp \De$, and 
called an  {\em $\lang$-quasiequation} if $|\De| = 1$ and an
{\em $\lang$-negative clause} if $\De = \emptyset$.
As usual, if the language is clear from the context, we may omit the prefix $\lang$.

Let us fix $\K$ to be a class of $\lang$-algebras, noting that  often in what follows $\K$ will consist of a finite set of 
$\lang$-algebras $\A_1,\ldots,\A_n$, and in this case we typically omit brackets. 
Given a finite  set of $\lang$-equations $\Si \cup \De$, 
we write $\Si \mdl{\K} \De$ and say that the $\lang$-clause $\Si \imp \De$ is {\em $\K$-valid}, 
if for every $\A \in \K$ and homomorphism $h \colon \alg{Tm_\lang} \to \A$,
\[
\Si \subseteq \ker h \qquad {\rm implies} \qquad \De \cap \ker h \neq \emptyset,
\]
recalling that $\ker h = \{(\f,\p)  :  h(\f) = h(\p)\}$. We also say that $\Si$ is {\em $\K$-satisfiable} if 
 $\Si \subseteq \ker{h}$ for some $\A \in \K$ and homomorphism $h \colon \alg{Tm_\lang} \to \A$. 

The class $\K$ is said to be {\em axiomatized} by a set of $\lang$-clauses $\Lambda$ if $\K$ is the class of $\lang$-algebras $\A$ 
such that all $\lang$-clauses in $\Lambda$ are $\A$-valid. $\K$ is called a {\em variety}, {\em quasivariety}, or 
{\em antivariety} if it is axiomatized by a set of $\lang$-equations, $\lang$-quasiequations, or 
$\lang$-negative clauses, respectively. The variety $\Vg(\K)$, quasivariety $\Qg(\K)$, and antivariety 
$\nVg(\K)$ {\em generated by} $\K$ are the smallest variety, quasivariety, and antivariety containing $\K$, respectively. 

Moreover, let $\op{H}$, $\op{I}$, $\op{S}$, $\op{P}$, $\op{P}_U$,  $\op{P}^*_U$, and $\op{H}^{-1}$ be the 
class operators of taking homomorphic images, isomorphic images, subalgebras, products, ultraproducts, non-empty ultraproducts, 
and homomorphic preimages, respectively. Then $\Vg(\K) = \op{HSP}(\K)$, $\Qg(\K) = \op{ISPP}_U(\K)$, and 
$\nVg(\K) = \op{H}^{-1}\op{S}\op{P}^*_U(\K)$, and if $\K$ is a finite set of finite algebras, these last two  
equivalences refine to  $\Qg(\K) = \op{ISP}(\K)$ and $\nVg(\K) = \op{H}^{-1}\op{S}(\K)$ (see~\cite{BS81,Gor98} for details).

Given a language $\lang$ and a set of variables $X$ such that either $X \neq \emptyset$ or $\lang$ contains at least 
one constant symbol, the term algebra $\alg{Tm_\lang}(X)$ exists and admits a congruence:
\[
\theta_\K(X) = \bigcap \{ \phi \in \Con(\alg{Tm_\lang}(X)) : \alg{Tm_\lang}(X)/\phi \in \op{IS}(\K) \}.
\]
Following \cite{BS81}, we let $\overline{X} = X / \theta_\K(X)$ and define the \emph{free algebra of $\K$ over $\overline{X}$}:
\[
\F_\K(\overline{X}) = \alg{Tm_\lang}(X) / \theta_\K(X).
\]
Then $\F_\K(\overline{X})$ has the universal mapping property for $\K$ over $\overline{X}$: namely, for each $\A \in \K$, 
any map from $\overline{X}$ to $A$ extends to a homomorphism from  $\F_\K(\overline{X})$ to $\A$ 
(\cite[Theorem II.10.10]{BS81}). 
If also $\K \neq \emptyset$, then $\F_\K(\overline{X}) \in \op{ISP}(\K) \subseteq \Qg(\K)$ (\cite[Theorem II.10.12]{BS81}). 

Note that $\F_\K(\overline{X}) \cong \F_\K(\overline{Y})$ whenever $|X| = |Y|$ (also $|\overline{X}| = |X|$ 
if $\K$ contains at least one non-trivial algebra). Hence we may consider for each cardinal $\kappa$, the 
(unique up to isomorphism) {\em free algebra of $\K$ on $\kappa$ generators} $\F_\K(\kappa)$, where $\F_\K(\kappa_1)$ 
is a subalgebra of $\F_\K(\kappa_2)$ for cardinals $\kappa_1 \le \kappa_2$. 
We note also that $\Vg(\K) = \Vg(\F_\K(\omega))$ and that 
$\Vg(\K_1) = \Vg(\K_2)$ implies $\F_{\K_1}(\omega) = \F_{\K_2}(\omega)$ 
(see~\cite[Corollary II.11.10; Exercise II.11.2]{BS81}). 


 \section{Finitely Generated Quasivarieties} \label{s:fingen}

A quasivariety $\Q$ is said to be {\em finitely generated} if $\Q = \Qg(\K)$ for some finite (generating) 
set $\K$ of finite $\lang$-algebras. Our first goal in this section will be to define a reasonable measure for 
comparing these generating sets; we will then  describe a method for obtaining a ``smallest'' generating set 
according to this measure. More precisely, we apply the standard multiset well-ordering defined by Dershowitz and 
Manna in~\cite{DeMa79} to the multiset of the cardinalities of algebras in $\K$. We  show that by decomposing 
finite $\lang$-algebras into their $\Q$-subdirectly irreducible components and appropriately 
refining the set of algebras obtained, we arrive at a (unique up to isomorphism) 
minimal set of $\lang$-algebras that still generates $\Q$.

Recall  that a \emph{multiset} over a set $S$ is an ordered pair $\langle S,f \rangle$ where
 $f$ is a function $f\colon S \to \mathbb{N}$, and is called \emph{finite} if $\{ x \in S  :  f(x) > 0 \}$ is finite.
For a well-ordered set $\langle S,\leq \rangle$, the 
\emph{multiset ordering $\leq_m$} on the set $M(S)$ of finite multisets over $S$ is defined by
$\langle S,f \rangle \leq_m \langle S,g \rangle$ if $f(x) > g(x)$ implies that for some $y \in S$, 
$y > x$ and $g(y) > f(y)$. It then follows that $\leq_m$ 
is a well-ordering of $M(S)$~(see~\cite{DeMa79}). As usual, we write a finite multiset of 
elements from $S$ as $[a_1,\ldots,a_n]$ where $a_1,\ldots,a_n \in S$ may include repetitions.

A  set of finite $\lang$-algebras $\{ \A_1, \dots, \A_n \}$ will be called a \emph{minimal generating set}
for the quasivariety $\Qg(\A_1, \dots, \A_n)$ if for every set of finite $\lang$-algebras
$\{ \B_1, \dots, \B_k \}$:
\[
\Qg(\A_1, \dots, \A_n) = \Qg(\B_1, \dots, \B_k) \qquad {\rm implies} \qquad 
[|A_1|, \dots, |A_n|] \leq_m [|B_1|, \dots, |B_k|].
\]
There are of course many measures that could be used to compare 
generating sets. Indeed, it may seem unreasonable to view  one hundred algebras with three elements as 
an improvement on a single algebra with four elements. However, the order is optimal in the 
following sense. In general, checking a quasiequation with $r$ variables in a finite algebra $\A$ 
requires checking $|A|^r$ assignments of variables to elements of $\A$. 
But then checking validity in  $\{ \A_1, \dots, \A_n \}$ will involve checking fewer assignments of variables than 
checking validity in $\{ \B_1, \dots, \B_k \}$ if $[|A_1|, \dots, |A_n|] \leq_m [|B_1|, \dots, |B_k|]$ 
 for quasiequations with sufficiently many variables.

To obtain minimal  generating sets  for $\Qg(\K)$ where $\K$ is a finite set of finite algebras, we 
 consider representations of the algebras in $\K$ using smaller algebras of $\Qg(\K)$. 
An algebra $\A$ is called a {\em subdirect product} of algebras $\A_1,\ldots,\A_n$ if there 
exist surjective homomorphisms $f_i\colon \A \to \A_i$ for $i = 1 \ldots n$ 
such that the induced homomorphism
\[
f\colon \A \to \A_1 \times \dots \times \A_n, \qquad f(x) = \langle f_1(x), \dots, f_n(x) \rangle
\]
is an embedding. In this case, $f$ is called a \emph{subdirect representation of $\A$}  and 
$\A_1,\ldots,\A_n$ are called {\em subdirect components} (for this representation) of $\A$. 
If $\Q$ is a quasivariety and $\A_1,\ldots, \A_n \in \Q$, then $\A$ is called  a {\em $\Q$-subdirect product of 
$\A_1, \dots, \A_n$}  and $f$ is called  a \emph{$\Q$-subdirect embedding}. 
$\A$ is called \emph{$\Q$-subdirectly irreducible} if for every $\Q$-subdirect embedding
$f\colon \A \to \A_1 \times \dots \times \A_n$, $\A$ is isomorphic to  $\A_i$ for some $i \in \{1,\ldots,n\}$. 

We will make essential use of the following:

\begin{lem}[{\cite[Corollary~6]{XC80}}] \label{lem:Qdecomp}
Let $\Q$ be a quasivariety and $\A \in \Q$.
Then $\A$ is a $\Q$-subdirect product of $\Q$-subdirectly irreducible members of $\Q$.
\end{lem}

\begin{lem}[{\cite[Proposition 3.1.6]{Gor98}}] \label{lem:irr_emb}
If $\Q = \Qg(\K)$ for a finite set $\K$ of finite algebras 
and $\A$ is a $\Q$-subdirectly irreducible algebra, then $\A \in \op{IS}(\K)$.
\end{lem}

\noindent
Observe now that if $\Q$ is a quasivariety and $\A \in \Q$  is a $\Q$-subdirect product of $\A_1, \dots, \A_n$, 
then each $\A_i$ is in $\Q$ by definition, and $\A$ is isomorphic to a subalgebra of a 
product of members of $\Q$. Hence we obtain:

\begin{lem} \label{lem:ReduceGens}
Let $\K$ be a class of $\lang$-algebras and suppose that $\K'$ is obtained from $\K$ by either (a) replacing 
$\A \in \K$ with $\A_1,\ldots,\A_n$ where $\A$ is a $\Qg(\K)$-subdirect product of $\A_1, \dots, \A_n$, 
or (b) replacing $\A, \B \in \K$ with $\B$ where $\A \in \op{IS}(\B)$. Then $\Qg(\K) = \Qg(\K')$.
\end{lem}

\noindent
In particular,  replacing each algebra $\A$ in a finite set $\K$ of finite algebras with the $\Qg(\K)$-subdirectly irreducible 
algebras in some $\Qg(\K)$-subdirect representation of $\A$, then removing any algebra that embeds into 
another algebra in the set, produces  a minimal generating set for $\Qg(\K)$ that is unique up to isomorphism.

\begin{thm} \label{t:MinGenSet}
Suppose that $\Q = \Qg(\A_1, \dots, \A_n )$ where  $\A_i$ is a  finite $\Q$-subdirectly irreducible algebra for  $i \in \{1 \dots n\}$ 
and $\A_i \not \in \op{IS}(\A_j)$ for $j \not = i$. Then $\{ \A_1, \dots, \A_n \}$ is the unique minimal generating set for $\Q$ 
up to isomorphism.
\end{thm}
\begin{proof}
Let $\Q = \Qg(\A_1, \dots, \A_n )$ where  $\A_i$ is a  finite $\Q$-subdirectly irreducible algebra for  $i \in \{1 \dots n\}$ 
and $\A_i \not \in \op{IS}(\A_j)$ for $j \not = i$. Suppose for a contradiction that $\Q = \Qg(\B_1,\ldots,\B_k)$ and 
$[|B_1|, \dots, |B_k|] <_m [ |A_1|, \dots, |A_n|]$. Without loss of generality, we can suppose that 
$\B_j$ is $\Q$-subdirectly irreducible for $j \in \{1, \ldots, k\}$; otherwise, by Lemmas~\ref{lem:Qdecomp} and~\ref{lem:ReduceGens}, 
$\B_j$ can be replaced with the $\Q$-subdirectly irreducible components of a $\Q$-subdirect representation 
of $\B_j$ and we obtain a smaller (according to $\le_m$) generating set of algebras for $\Q$. 

It follows that there exists a largest $r \in \mathbb{N}$ such that  there are strictly more occurrences of 
$r$ in $[ |A_1|, \dots, |A_n|]$ than in $[|B_1|, \dots, |B_k|]$, and 
for each $r' > r$, the number of occurrences of $r'$ in $[ |A_1|, \dots, |A_n|]$ and 
$[|B_1|, \dots, |B_k|]$ are equal. Each $\A_i$ is finite and 
$\Q$-subdirectly irreducible, and hence by Lemma~\ref{lem:irr_emb}, 
embeds into some $\B_j$ where $|A_i| \le |B_j|$. If every $\A_i$ of size $r$ 
embeds into, and is hence isomorphic to, a $\B_j$ of size $r$, 
then (by the pigeonhole principle) there must be two isomorphic algebras in $\{ \A_1,\ldots,\A_n \}$, a contradiction. 
Hence, suppose without loss of generality that $\A_1$ embeds into $\B_1$ with $|A_1| = r$ and 
$|B_1| > r$. But notice now that $\B_1$ is also $\Q$-subdirectly irreducible and hence embeds into some  
$\A_i$ with $i \in \{2,\ldots,n\}$. So  $\A_1 \in \op{IS}(\A_i)$, a contradiction.

Finally, consider any minimal generating set $\{ \B_1,\ldots,\B_k \}$ for $\Q$, and suppose for a contradiction 
that $\B_i \not \in \op{I}(\A_1,\ldots,\A_n)$ for some $i \in \{1,\ldots,k\}$. 
Then  by Lemma~\ref{lem:irr_emb}, $\B_i$ properly embeds into  
$\A_j$ for some  $j \in \{1,\ldots,n\}$. But also  by Lemma~\ref{lem:irr_emb}, $\A_j$ embeds into 
$\B_d$ for some $d \in \{1,\ldots,k\} \setminus \{i\}$. It follows that $\B_i$ can be embedded into the strictly larger 
algebra $\B_d$. But then $\{ \B_1,\ldots,\B_k \}$ is not a minimal generating set for $\Q$, a 
contradiction. 
\end{proof}

In the remainder of this section, we develop some results for $\Q$-subdirect products and $\Q$-subdirectly irreducible algebras, and 
present an algorithm for obtaining minimal generating sets for finitely generated quasivarieties.
First, let us recall the following theorem of Birkhoff, which establishes a useful relationship between subdirect representations of
a given algebra $\A$ and sets of congruences on $\A$.

\begin{lem}[{\cite[Universal~Algebra,~Theorem~11]{Bir40}}] \label{lem:correspondence}
If $\A$ is a subdirect product of the family  $(\A_i)_{i \in I}$,
then there exist for $i \in I$, congruences $\th_i \in \Con(\A)$ such that 
$\A_i \cong \A/\th_i$ and $\bigcap_{i \in I} \th_i = \De_\A$.
Conversely, for any family of congruences $(\th_i)_{i \in I}$ on $\A$, the quotient
$\A/(\bigcap_{i \in I} \th_i)$ is a subdirect product of the family $(\A/\th_i)_{i \in I}$.
\end{lem}

\noindent
The set of \emph{$\Q$-congruences on $\A$} is defined as $\Con_\Q(\A) = \{ \th \in \Con(\A) : \A/\th \in \Q \}$. 
Clearly, the above lemma also holds for $\Q$-subdirect representations of an algebra $\A$ with 
respect to $\Q$-congruences. 
Note, moreover, that the number of congruences needed to obtain a 
subdirect representation of a finite algebra $\A$ is at most $|A|$, the maximal number of coatoms 
of the congruence lattice $\Con(\A)$.

\begin{cor} \label{cor:correspondence}
Let $\Q$ be a quasivariety and $\A \in \Q$.

\begin{enumerate}[\rm(a)]
\item

If $\A$ is a $\Q$-subdirect product of  the family $(\A_i)_{i \in I}$,
 then there exist for $i \in I$, $\Q$-congruences $\th_i \in \Con_{\Q}(\A)$ such that 
$\A_i \cong \A/\th_i$ and $\bigcap_{i \in I} \th_i = \De_\A$.
Conversely, for any family of $\Q$-congruences $(\th_i)_{i \in I}$ on $\A$ with 
$\bigcap_{i \in I} \th_i = \De_\A$, $\A$ is a $\Q$-subdirect product  of the family $(\A/\th_i)_{i \in I}$.

\item
$\A$ is $\Q$-subdirectly irreducible iff the bottom element $\De_\A$ of $\Con_\Q(\A)$ is meet-irreducible 
(i.e., if $\bigcap_{i \in I} \th_i = \De_\A$ for $\th_i \in \Con_\Q(\A)$, then $\De_\A = \th_i$ for some $i \in I$).
\end{enumerate}
\end{cor}

\noindent The problem of finding the congruence closure for a given equivalence relation on a finite algebra,
i.e., the smallest congruence containing this equivalence, can be solved in polynomial time.
This result was used in \cite{Dem82} to provide a polynomial time algorithm for calculating a subdirect 
representation of a finite algebra. The problem of finding the $\Q$-congruence closure of an equivalence relation on a finite 
algebra with respect to a finitely generated quasivariety $\Q$ appears to be much harder, however. Instead,  
we use here the following characterization of $\Q$-subdirectly irreducible algebras as the basis for a procedure that constructs 
$\Q$-subdirectly irreducible components for a $\Q$-subdirect representation of a given finite algebra without needing to
calculate the $\Q$-congruence lattice.

\begin{lem} \label{lem:CharQSubDir}
For a finite set $\K$ of finite algebras and $\A \in \Qg(\K)$, the following are equivalent:
\begin{enumerate}[\rm (1)]
\item $\A$ is $\Qg(\K)$-subdirectly irreducible.
\item $\bigcap \{ \th \in \Con(\A)\setminus\{\De_\A\} : \A/\th \in \op{IS}(\K) \} \neq \De_\A$.
\end{enumerate}
\end{lem}
\begin{proof}
For convenience, let
\[
\Theta = \{ \th \in \Con(\A)\setminus\{\De_\A\} : \A/\th \in \op{IS}(\K) \} \subseteq \Con_{\Qg(\K)}(\A).
\]
(1)$\Rightarrow$(2) We proceed contrapositively. If $\bigcap \Theta = \De_\A$, then by Corollary~\ref{cor:correspondence}(a), 
$\A$ is a $\Qg(\K)$-subdirect product of algebras in $\{\A / \th : \th \in \Theta\}$.  
But also by Corollary~\ref{cor:correspondence}(b),  since $\De_\A \not\in \Theta$, $\A$ is not $\Qg(\K)$-subdirectly irreducible. \medskip

\noindent(2)$\Rightarrow$(1) Again, we proceed contrapositively.  
If $\A$ is not $\Qg(\K)$-subdirectly irreducible, then combining  Lemma~\ref{lem:Qdecomp} and 
Corollary~\ref{cor:correspondence}, there exist $(\th_i)_{i \in I} \subseteq \Con_{\Qg(\K)}(\A)\setminus\{\De_\A\}$ such that 
$\bigcap_{i \in I} \th_i = \De_\A$ and $\A$ is a $\Qg(\K)$-subdirect product of $\Qg(\K)$-subdirectly irreducible 
algebras $\A / \th_i$ ($i \in I$). But then also  by Lemma~\ref{lem:irr_emb}, we have 
$\A/\th_i \in \op{IS}(\K)$ for each $i \in I$. So $(\th_i)_{i \in I} \subseteq \Theta$,
and hence $\bigcap \Theta = \De_\A$.
\end{proof}

\begin{figure}[t]
\begin{algorithmic}[1]
\Function {MinGenSet}{$\K$}
   \State \textbf{declare} $S_1, S_2, C$ : set
   \State \textbf{declare} $\mathcal{M}$ : list
   \State \textbf{declare} $\A$ : algebra
   \State \textbf{declare} $i$ : integer
   \State $\mathcal{M} \leftarrow \mathrm{list}(\K)$

   \State $i \leftarrow 1$
   \While{$i \leq \mathrm{length}({\mathcal{M}})$}
      \State $\A \leftarrow {\mathcal{M}}[i]$
      \State $C \leftarrow \Con(\A)\setminus\{\De_\A\}$ \label{alg:MinGenSet:Con}
      \State $S_1 \leftarrow \{ \th \in C : \A/\th$ embeds into $\A \}$ \label{alg:MinGenSet:S1}
      \State $S_2 \leftarrow \{ \th \in C : \A/\th$ embeds into some ${\mathcal{M}}[j] \neq \A \}$
      \If{$\bigcap (S_1 \cup S_2) = \De_\A$}
      {
         \ForAll{$\th$ in $S_1 \setminus S_2$}
            \State add $\A/\th$ to ${\mathcal{M}}$ \label{alg:MinGenSet:AddQuot}
         \EndFor
      }
         \State remove $\A$ from ${\mathcal{M}}$ \label{alg:MinGenSet:DelA}
      \Else
         \State $i \leftarrow i + 1$
      \EndIf
   \EndWhile
   \ForAll{$\A$ in ${\mathcal{M}}$} \label{alg:MinGenSet:E1}
      \If{$\A$ embeds into some ${\mathcal{M}}[j] \neq \A$}
         \State remove $\A$ from ${\mathcal{M}}$
      \EndIf
   \EndFor \label{alg:MinGenSet:E2}
   \State \Return {set($\mathcal{M}$)}
\EndFunction
\end{algorithmic}
\caption{For a finite set $\K$ of finite algebras, return the minimal generating set of $\Qg(\K)$. \label{alg:MinGenSet}}
\end{figure}

We now have all the  ingredients necessary to describe an algorithm \proc{MinGenSet} 
(see Figure~\ref{alg:MinGenSet}) that calculates the  (unique up to isomorphism) minimal generating set 
for a quasivariety $\Q = \Qg(\K)$, where $\K$ is a finite set of finite $\lang$-algebras. 
By Theorem~\ref{t:MinGenSet}, it suffices to find a set of $\Q$-subdirectly irreducible algebras that 
generates $\Q$, where no member of the set embeds into another member of the set. 

The algorithm proceeds by considering each $\A \in \K$ in turn. First, the congruence lattice $\Con(\A)$ 
is generated  (line~\ref{alg:MinGenSet:Con}) by checking for all equivalence relations if they are congruences. 
Next, the congruences $\th \in \Con(\A) \setminus \{\De_\A\}$ such that $\A / \th$ 
embeds into $\A$ or some other member of $\K$ are collected in sets $S_1$ and $S_2$, respectively. 
If  $\bigcap (S_1 \cup S_2) \not = \De_\A$, then $\A$ is $\Q$-subdirectly irreducible 
by Lemma~\ref{lem:CharQSubDir}, so the algorithm proceeds to the next algebra in $\K$. 
Otherwise $\bigcap (S_1 \cup S_2) = \De_\A$ and by Lemma~\ref{lem:CharQSubDir}, $\A$ is not 
$\Q$-subdirectly irreducible. In this case, for each $\th \in S_1 \setminus S_2$, the algebra 
$\A / \th$ is added to $\K$ (line~\ref{alg:MinGenSet:AddQuot}) and $\A$ is removed from $\K$ 
 (line~\ref{alg:MinGenSet:DelA}). Note that since the cardinalities of the added algebras are strictly smaller 
 than the cardinality of the removed algebra, the new set of algebras is smaller according to the multiset 
 ordering defined in Section~\ref{s:fingen}. Hence this procedure is terminating. Moreover, the resulting 
 finite set of finite algebras must generate the quasivariety $\Q$ (by~Lemma \ref{lem:ReduceGens}), contain 
 only $\Q$-subdirectly irreducible algebras, and not contain any algebra that embeds into another member of the set
(lines~\ref{alg:MinGenSet:E1}--\ref{alg:MinGenSet:E2}). 
 Hence by Theorem~\ref{t:MinGenSet}, we obtain:

\begin{thm}
For a finite set $\K$ of finite $\lang$-algebras, \proc{MinGenSet}($\K$)
returns the  (unique up to isomorphism) minimal generating set for the quasivariety $\Qg(\K)$.
\end{thm}

\noindent
We remark that although the algorithm \proc{MinGenSet} does not need to calculate the $\Q$-congruence lattice, already 
calculating the congruence lattice of a finite algebra can take exponential time. Moreover, we make frequent 
use in our algorithm here of checking embeddings, which is in general an NP-complete problem (see~\cite{Mil79, BS00}). 


\section{Admissibility and Free Algebras} \label{s:admfree}

Let us fix a class of $\lang$-algebras $\K$. A homomorphism $\si \colon \alg{Tm}_\lang \rightarrow \alg{Tm}_\lang$ 
is called a \emph{$\K$-unifier} of a set $\Si$ of $\lang$-equations  if $\mdl{\K} \si(\f) \eq \si(\p)$ for every $(\f \eq \p) \in \Si$; 
in this case, $\Si$ is said to be {\em $\K$-unifiable}. An $\lang$-quasiequation 
$\Si \imp \f \eq \p$ is {\em $\K$-admissible}  if every $\K$-unifier of $\Si$ is a $\K$-unifier of $\f \eq \p$.  
In this section, we present characterizations of $\K$-unifiability, $\K$-admissibility, and related properties,  
emphasizing their close relationship to free algebras. The characterizations obtained will then be used 
in Section~\ref{s:algorithms} to develop algorithms for checking these properties in the context of finitely generated 
quasivarieties.

Let us first take a closer look at $\K$-unifiability, noting that a finite set $\Si$ of $\lang$-equations is $\K$-unifiable 
if and only if the $\lang$-negative clause $\Si \imp \emptyset$ is not $\K$-admissible
(equivalently, when $\K$ contains a non-trivial algebra, if and only if 
the $\lang$-quasiequation $\Si \imp x \eq y$ with $x,y$ not occurring in $\Si$ is not $\K$-admissible). 
We will see that for checking $\K$-unifiability, the ``optimal'' solution is to check satisfiability in the 
smallest finite subalgebra (if such an algebra exists) of $\F_\K(\omega)$. 

\begin{lem}\label{l:unif}
The following are equivalent for any class $\K'$ of $\lang$-algebras:
\begin{enumerate}[\rm (1)]
\item $\Si$ is $\K$-unifiable \ iff \ $\Si$ is $\K'$-satisfiable.
\item $\nVg(\K') = \nVg(\F_{\K}(\omega))$.
\end{enumerate}
\end{lem}
\begin{proof}
Recall that $\nVg(\K') = \nVg(\F_{\K}(\omega))$ is equivalent to the condition that  an 
$\lang$-negative clause $\Si \imp \emptyset$ is $\K'$-valid iff it  is $\F_\K(\omega)$-valid. 
However,  $\Si \imp \emptyset$ is  $\K'$-valid iff $\Si$ is not $\K'$-satisfiable and 
$\Si \imp \emptyset$ is $\F_\K(\omega)$-valid iff $\Si$ is not $\F_\K(\omega)$-satisfiable. 
For the equivalence of (1) and (2), it suffices therefore to show that 
 $\Si$ is $\F_\K(\omega)$-satisfiable iff $\Si$ is $\K$-unifiable. Suppose first that 
$h \colon \alg{Tm_\lang} \to \F_\K(\omega)$ satisfies $\Si$. Then any homomorphism 
$\si \colon \alg{Tm_\lang} \to  \alg{Tm_\lang}$ defined such that $\si(x) \in h(x)$ for each variable $x$ 
is a $\K$-unifier of $\Si$. Conversely, if $\sigma$ is a $\K$-unifier of $\Si$, then the homomorphism 
$h \colon \alg{Tm_\lang} \to \F_\K(\omega)$ defined by $h(x) = \sigma(x) / \theta_\K(\omega)$  for each variable $x$ 
satisfies $\Si$. 
\end{proof}

\begin{prop}\label{p:unif}
Let $\K$ be a class of $\lang$-algebras and $\C \in \op{S}(\F_\K(\omega))$. 
\begin{enumerate}[\rm(a)]
\item	$\Si$ is $\K$-unifiable \ iff \ $\Si$ is $\C$-satisfiable.
\item	If $\C$ is a smallest finite subalgebra of $\F_\K(\omega)$ and 
			$\K'$ is a class of $\lang$-algebras such that $\Si$ is $\K$-unifiable iff $\Si$ is $\K'$-satisfiable, 
			then $|C| \le |B|$ for each $\B \in \K'$.
\end{enumerate}
\end{prop}
\begin{proof}
(a) By assumption, $\C \in \nVg(\F_\K(\omega))$, so $\nVg(\C) \subseteq  \nVg(\F_\K(\omega))$. 
But also, since $\C \in \op{S}(\F_\K(\omega)) \subseteq \Vg(\F_\K(\omega)) = \Vg(\K)$ and $\F_\K(\omega) = \F_{\Vg(\K)}(\omega)$ has the universal mapping property for $\Vg(\K)$ 
over countably infinitely many generators, we obtain a homomorphism $h \colon \F_\K(\omega) \to \C$ 
defined by $h(x) = c$ for every variable $x$ for some fixed $c \in C$. Hence $h[\F_\K(\omega)]$ is a subalgebra of $\C$ and  
$\F_\K(\omega) \in \op{H}^{-1}\op{S}(\C) \subseteq \nVg(\C)$. So 
$\nVg(\F_\K(\omega)) \subseteq \nVg(\C)$ and the result follows by Lemma~\ref{l:unif}.

(b) Let $\C$ be a smallest finite subalgebra of $\F_\K(\omega)$ and suppose that 
$\K'$ is a class of $\lang$-algebras such that $\Si$ is $\K$-unifiable iff $\Si$ is $\K'$-satisfiable. 
Then by Lemma~\ref{l:unif} and part (a),  $\nVg(\K') = \nVg(\F_{\K}(\omega)) = \nVg(\C)$. 
Hence if $\B \in \K' \subseteq \nVg(\K') = \nVg(\F_{\K}(\omega)) = \nVg(\C) = \op{H}^{-1}\op{S}\op{P}^*_U(\C) = \op{H}^{-1}(\C)$,
then clearly $|C| \le |B|$.
\end{proof}

\begin{exa}\label{ex:DeMorganUni}
The variety $\vty{DMA}$ of De Morgan algebras is generated as a quasivariety by the $4$-element  algebra 
$\alg{D_4} = \langle \{ \bot,a,b,\top\}, \land, \lor, \lnot, \bot, \top \rangle$ 
consisting of a distributive bounded lattice with an involutive negation 
defined as follows:
   
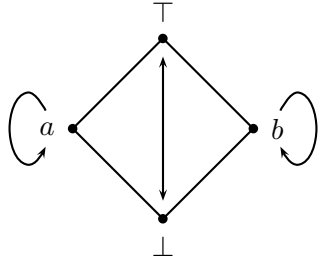
\begin{center}  
   \pspicture(-4,-0.5)(4,3)
   \psset{unit=12mm}
   \psline(0,0)(1,1)
   \psline(1,1)(0,2)
   \psline(0,0)(-1,1)
   \psline(-1,1)(0,2)
   \psdots(0,0)(-1,1)(1,1)(0,2)
   \rput[m](0,-0.3){$\bot$}
   \rput[r](-1.2,1){$a$}
   \rput[l](1.2,1){$b$}
   \rput[m](0,2.3){$\top$}
   \pscurve{->}(1.3,1.2)(1.5,1.4)(1.7,1)(1.5,0.6)(1.3,0.8)
   \pscurve{->}(-1.3,1.2)(-1.5,1.4)(-1.7,1)(-1.5,0.6)(-1.3,0.8)
   \psline{<->}(0,0.2)(0,1.8)
   \endpspicture
   \end{center}
Since there are constants in the language of $\alg{D_4}$, the smallest algebra for checking $\vty{DMA}$-unifiability is 
the $2$-element ground algebra $\F_{\alg{D_4}}(0)$: i.e., the $2$-element Boolean algebra. That is, checking unifiability 
amounts to checking classical satisfiability. E.g., $x \land \lnot x \eq x \lor \lnot x$ is not $\vty{DMA}$-unifiable, since 
in the $2$-element Boolean algebra, $\top \land \lnot \top \neq \top \lor \lnot \top$ and $\bot \land \lnot \bot \neq \bot \lor \lnot \bot$. 
The case of the ``constant-free'' variety $\vty{DML}$ of De Morgan lattices, generated as 
a quasivariety by $\alg{D^{\ell}_4} = \langle \{ \bot,a,b,\top\}, \land, \lor, \lnot \rangle$, 
is not so immediate. However, there is also a smallest $2$-element subalgebra of $\F_{\alg{D^{\ell}_4}}(\omega)$ with elements 
corresponding to $x \land \lnot x$ and $x \lor \lnot x$. So checking  $\vty{DML}$-unifiability 
amounts again to checking classical satisfiability. 
\end{exa}

Let us turn our attention now to admissible quasiequations, recalling the following useful characterizations of admissibility 
in terms of free algebras and generated subvarieties.

\begin{thm}[{\cite[Theorem 2]{CM201x}}; see also {\cite[Theorem 1.4.5]{Ryb97}}] \label{t:EqAdm} 
The following are equivalent:
\begin{enumerate}[{\rm (1)}]
\item		$\Si \imp \f \eq \p$ is $\K$-admissible.
\item		$\Si \imp \f \eq \p$ is $\Qg(\K)$-admissible.
\item		$\Si \mdl{\F_{\K}(\omega)} \f \eq \p$.
\item		$\Vg(\K) = \Vg( \{\A \in \Qg(\K)  :  \Si \mdl{\A} \f \eq \p\})$. 
\end{enumerate}
\end{thm}

\begin{exa}\label{ex:lattices}
The following quasiequations, expressing meet and join semi-distributivity in the language of lattices with $\land$ and $\lor$, 
are satisfied by all free lattices (see~\cite{Jon61}), and are therefore admissible in the variety of lattices:
\[
\begin{array}{rcl}
\{x \land y \eq x \land z\} & \imp & x \land y \eq x \land (y \lor z)\\[.05in]
\{x \lor y \eq x \lor z\} & \imp & x \lor y \eq x \lor (y \land z).
\end{array}
\]
Similarly, the quasiequations below expressing torsion-freeness in a language of groups with 
$\cdot$, $^{-1}$, and ${\rm e}$ are satisfied by  all free groups, and are hence admissible in the variety of groups:
\[
\{\underbrace{x \cdot \ldots \cdot x}_n \eq {\rm e}\} \quad \imp \quad x \eq {\rm e} \qquad (n = 2,3,\ldots).
\]
\end{exa}

\noindent
Given a class $\K$ of $\lang$-algebras, we will be interested in determining when the $\K$-admissibility of quasiequations 
coincides with their $\K'$-validity in another class of $\lang$-algebras $\K'$. By Theorem~\ref{t:EqAdm}, this is the case exactly 
when  $\Qg(\K') = \Qg(\F_{\K}(\omega))$. The next result provides a further useful characterization of this situation.

\begin{prop} \label{p:adm}
The following are equivalent:
\begin{enumerate}[\rm (1)]
\item $\Si \imp \f \eq \p$ is $\K$-admissible \ iff \ $\Si \mdl{\K'} \f \eq \p$.
\item $\Qg(\K') = \Qg(\F_{\K}(\omega))$.
\item $\K' \subseteq \Qg(\F_\K(\omega))$ and $\K \subseteq \Vg(\K')$.
\end{enumerate}
\end{prop}
\begin{proof}
(1) $\Leftrightarrow$ (2) Follows directly from Theorem~\ref{t:EqAdm}. \medskip 

\noindent(2) $\Rightarrow$ (3) Suppose that $\Qg(\K') = \Qg(\F_{\K}(\omega))$. Then 
$\K' \subseteq \Qg(\F_\K(\omega))$. Moreover, 
$\Vg(\K') = \Vg(\Qg(\K')) = \Vg(\Qg(\F_{\K}(\omega))) =  \Vg(\F_{\K}(\omega)) = \Vg(\K)$, 
so  $\K \subseteq \Vg(\K')$.\medskip

\noindent(3) $\Rightarrow$ (2) Suppose that $\K' \subseteq \Qg(\F_\K(\omega))$ and $\K \subseteq \Vg(\K')$. 
Then clearly $\Qg(\K') \subseteq \Qg(\F_\K(\omega))$. But also 
$\Vg(\K) \subseteq \Vg(\K') \subseteq \Vg(\Qg(\F_\K(\omega))) = \Vg(\F_\K(\omega)) = \Vg(\K)$. 
That is, $\Vg(\K) = \Vg(\K')$. Hence $\F_\K(\omega) = \F_{\K'}(\omega) \in \Qg(\K')$ and 
$\Qg(\F_\K(\omega)) \subseteq \Qg(\K')$. 
\end{proof}

\noindent
For some well-behaved classes of algebras, admissibility and validity coincide. More precisely, 
a class $\K$ of $\lang$-algebras is said to be {\em structurally complete} if it satisfies the condition:
\[
\Si \imp \f \eq \p \textrm{ is $\K$-admissible} \qquad \mbox{iff} \qquad \Si \mdl{\K} \f \eq \p.
\]
The following characterization then follows almost immediately from Theorem~\ref{t:EqAdm}:

\begin{prop}[{\cite[Proposition 2.3]{Ber91}}]\label{p:Bergman}
The following are equivalent:
\begin{enumerate}[\rm (1)]
\item $\K$ is structurally complete.
\item $\Qg(\K) = \Qg(\F_\K(\omega))$.
\item  $\K' \subseteq \K$ and $\Vg(\K') = \Vg(\K)$ implies $\Qg(\K') = \Qg(\K)$.
\end{enumerate}
\end{prop}

\begin{exa}
The variety $\vty{BA}$ of Boolean algebras, generated as a quasivariety by the $2$-element 
algebra $\alg{2} = \langle \{0,1\}, \land, \lor, \lnot, 0, 1 \rangle$, is structurally complete. It suffices to observe that 
$\alg{2}$ embeds into $\F_{\alg{2}}(\omega) = \F_{\vty{BA}}(\omega)$ via a mapping that sends 
$0$ and $1$ to their respective equivalence classes, and hence that 
$\vty{BA} = \Qg(\alg{2})  = \Qg(\F_{\vty{BA}}(\omega))$.
\end{exa}

\begin{exa} \label{ex:M5N5}
A \emph{modular lattice} $\alg{L}$ may be characterized as a lattice satisfying the equation
$(x \land y) \lor (y \land z) \eq y \land ((x \land y) \lor z)$. Famously, 
a lattice $\alg{L}$ is non-modular if and only if the lattice $\alg{N_5}$ (below) 
embeds into $\alg{L}$ (see~\cite[Theorem~I.3.5]{BS81}). 
But since $\alg{N_5}$ is non-modular, also $\F_{\alg{N_5}}(\omega)$ (which must satisfy the same equations) is non-modular. 
So $\alg{N_5}$ embeds into $\F_{\alg{N_5}}(\omega)$, and $\Qg(\alg{N_5})$ is structurally complete. 
Similarly, it is well-known that a lattice $\lgc{L}$ is distributive if and only if 
neither $\alg{N_5}$ nor $\alg{M_5}$ (below) embeds into $\alg{L}$  (see~\cite[Theorem~I.3.6]{BS81}). 
Since $\alg{M_5}$ is non-distributive and modular, also $\F_{\alg{M_5}}(\omega)$ is non-distributive and modular. 
So $\alg{M_5}$ embeds into $\F_{\alg{M_5}}(\omega)$, and $\Qg(\alg{M_5})$ is structurally complete. 

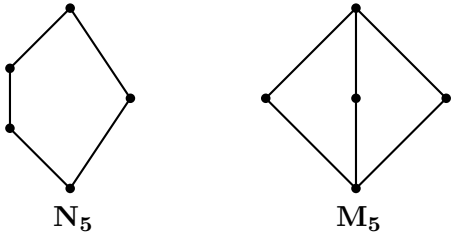
\begin{center}  
\begin{pspicture}(0,-2)(5.94,1.5)
\psdots(0.86,1.2)
\psdots(0.06,0.4)
\psdots(0.06,-0.4)
\psdots(0.86,-1.2)
\psdots(1.66,0.0)
\psdots(4.66,0.0)
\psdots(4.66,1.2)
\psdots(4.66,-1.2)
\psdots(5.86,0.0)
\psdots(3.46,0.0)
\psline(0.06,0.4)(0.86,1.2)
\psline(1.66,0.0)(0.86,1.2)
\psline(0.06,0.4)(0.06,-0.4)
\psline(0.06,-0.4)(0.86,-1.2)
\psline(0.86,-1.2)(1.66,0.0)
\psline(3.46,0.0)(4.66,-1.2)
\psline(4.66,-1.2)(5.86,0.0)
\psline(5.86,0.0)(4.66,1.2)
\psline(4.66,1.2)(3.46,0.0)
\psline(4.66,-1.2)(4.66,0.0)
\psline(4.66,0.0)(4.66,1.2)
\rput[m](0.9,-1.6){$\alg{N_5}$}
\rput[m](4.7,-1.6){$\alg{M_5}$}
\end{pspicture} 
\end{center}
\end{exa}

\noindent
For certain other classes, admissibility and validity coincide for quasiequations with unifiable premises. 
More precisely, we call a class $\K$ of $\lang$-algebras  {\em almost structurally complete} if  it satisfies the  condition:
\[
\Si \imp \f \eq \p \textrm{ is $\K$-admissible} \qquad\mbox{iff} 
\qquad \Si \mdl{\K} \f \eq \p \ \textrm{ or \ $\Si$ is not $\K$-unifiable.}
\]
The next result provides a useful characterization of this situation:

\begin{thm}\label{thm:ASC}
The following are equivalent for any $\B  \in \op{S}(\F_\K (\omega))$:
\begin{enumerate}[{\rm (1)}]
\item $\K$ is almost structurally complete.
\item $\Qg(\{\A \times \B  :  \A \in \K\}) = \Qg(\F_\K(\omega))$.
\item $\{\A \times \B  :  \A \in \K\} \subseteq \Qg(\F_\K(\omega))$.
\end{enumerate}
\end{thm}
\begin{proof} 
(1) $\Rightarrow$ (2) Suppose that $\K$ is almost structurally complete. 
To establish  $\Qg(\{\A \times \B  :  \A \in \K\}) = \Qg(\F_\K(\omega))$, it suffices to show 
that a quasiequation $\Si \imp \f \eq \p$ is valid in all algebras 
$\A \times \B$ for $\A \in \K$ iff it is valid in $\F_\K(\omega)$. 
Suppose first that  $\Si \mdl{\F_\K(\omega)} \f \eq \p$. Then by Theorem~\ref{t:EqAdm}, 
either $\Si$ is not $\K$-unifiable or $\Si \imp \f \eq \p$ is $\K$-valid. In the first case, by Proposition~\ref{p:unif}, 
$\Si$ is not $\B$-satisfiable, so $\Si \imp \f \eq \p$ is valid in $\A \times \B$ for all $\A \in \K$. 
In the second case,  $\Si \imp \f \eq \p$ is valid in $\A \times \B \in \Qg(\K)$ for all $\A \in \K$. Conversely, 
if $\Si \imp \f \eq \p$ is valid in $\A \times \B$ for each $\A \in \K$, then 
either $\Si$ is not $\B$-satisfiable 
or $\Si \imp \f \eq \p$ is valid in each $\A$ in $\K$. In the first case, by Proposition~\ref{p:unif}, $\Si$ is not $\K$-unifiable,
so $\Si \imp \f \eq \p$ is valid in $\F_\K (\omega)$.
In the second case, $\Si \imp \f \eq \p$ is valid in $\Qg(\K)$ and  hence valid in  $\F_\K (\omega)$. \medskip

\noindent(2) $\Rightarrow$ (1) Suppose that  $\Qg(\F_\K (\omega)) = \Qg(\{\A \times \B  :  \A \in \K\})$. Then whenever 
$\Si \imp \f \eq \p$ is $\K$-admissible,  it is  $\F_\K(\omega)$-valid and hence also valid in $\A \times \B$ for all 
$\A \in \K$. Moreover, if $\Si$ is $\K$-unifiable, then, by Proposition~\ref{p:unif}, it is $\B$-satisfiable. 
I.e., there exists a homomorphism $h  \colon \alg{Tm_\lang} \to \B$ with $\Si \subseteq \ker h$. 
 For any $\A \in \K$ and homomorphism $k  \colon\alg{Tm_\lang} \to \A$ with $\Si \subseteq \ker k$, define
$e_\A  \colon \alg{Tm_\lang} \to \A \times \B$ by $e_\A(u) = (k(u),h(u))$. Then,
since $\Si \imp \f \eq \p$ is valid in $\A \times \B$ for all $\A \in \K$, 
$\Si \subseteq \ker e$, so $e(\f) = e(\p)$ and  $k(\f) = k(\p)$. I.e., $\Si \mdl{\A} \f \eq \p$.  
So we have shown that $\Si \mdl{\K} \f \eq \p$. \medskip

\noindent(2) $\Rightarrow$ (3) Immediate. \medskip

\noindent(3) $\Rightarrow$ (2) Suppose that $\{\A \times \B  :  \A \in \K\} \subseteq \Qg(\F_\K(\omega))$. 
Then also, since $\A \in \op{H}(\A \times \B)$ for each $\A \in \K$, we obtain $\K \subseteq \Vg(\{\A \times \B  :  \A \in \K\})$. 
Hence by Proposition~\ref{p:adm}, $\Qg(\{\A \times \B  :  \A \in \K\}) = \Qg(\F_\K(\omega))$.
\end{proof}

\begin{exa} \label{ex:luk}
Consider the $2$-element and $3$-element Wajsberg algebras 
$\alg{L_2}=\langle \{0,1\}, \to, \lnot \rangle$ and 
$\alg{L_3} = \langle \{ 0, \frac{1}{2}, 1 \}, \to, \lnot \rangle$ where 
\[
x \to y = \min(1, 1-x+y) \quad \mbox{and} \quad \lnot x = 1 - x.
\] 
The algebra $\alg{L_3} \times \alg{L_2}$ embeds into $\alg{\F_{L_3}}(\omega)$, as illustrated in the diagram 
below by the terms associated to elements, and has $\alg{L_3}$ as  
a homomorphic image, as indicated by the arrows. Hence $\alg{L_3}$ is almost structurally complete. 
Note that it is not structurally complete since, for example, $\{x \eq \lnot x\} \imp \emptyset$ is 
$\alg{L_3}$-admissible, but not $\alg{L_3}$-valid. 
On the other hand, its implicational reduct $\alg{L_3^\to}= \langle \{ 0, \frac{1}{2}, 1 \}, \to \rangle$ is 
structurally complete, since it embeds into the free algebra 
$\alg{\F_{L_3^\to}}(2)$.\\
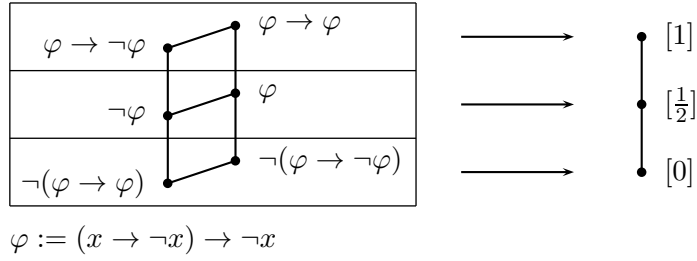
   \begin{center}
   \pspicture(-5,-1)(10,2.5)
   \psset{unit=3mm}
   \psdots(0,0)(3,1)(0,3)(3,4)(0,6)(3,7)(21,0.5)(21,3.5)(21,6.5)

   \rput[r](-1,0){$\lnot(\f \to \f)$}
   \rput[r](-1,3){$\lnot \f$}
   \rput[r](-1,6){$\f \to \lnot \f$}
   \rput[l](4,1){$\lnot(\f \to \lnot \f)$}
   \rput[l](4,4){$\f$}
   \rput[l](4,7){$\f \to \f$}
   \rput[l](22,0.5){$[0]$}
   \rput[l](22,3.5){$[\frac{1}{2}]$}
   \rput[l](22,6.5){$[1]$}

   \psline(0,0)(3,1)
   \psline(0,3)(3,4)
   \psline(0,6)(3,7)
   \psline(0,0)(0,6)
   \psline(3,1)(3,7)
   \psline(21,0.5)(21,6.5)

   \psline{->}(13,0.5)(18,0.5)
   \psline{->}(13,3.5)(18,3.5)
   \psline{->}(13,6.5)(18,6.5)

   \psline[linewidth=0.5pt](-7,-1)(11,-1)
   \psline[linewidth=0.5pt](-7,2)(11,2)
   \psline[linewidth=0.5pt](-7,5)(11,5)
   \psline[linewidth=0.5pt](-7,8)(11,8)
   \psline[linewidth=0.5pt](-7,-1)(-7,8)
   \psline[linewidth=0.5pt](11,-1)(11,8)
   \rput[l](-7,-2.5){$\f := (x \to \lnot x) \to \lnot x$}
   \endpspicture
   \end{center}
  \end{exa}

 
\section{Algorithms and Examples}~\label{s:algorithms}

In this section, we present algorithms for checking admissibility of quasiequations and related properties in  
finitely generated quasivarieties. We also present results on admissibility (some new, some well-known) for a range 
of examples taken from the universal algebra and non-classical logic literature. These results, obtained using a Delphi
implementation of the algorithms described here, are collected at the end of this section in Table~\ref{table}. 

We begin by recalling Birkhoff's result that the finitely generated free algebras for a finite set $\K$ of finite algebras are 
themselves finite; more precisely:

\begin{lem}[\cite{Bir35}]\label{l:finite}
For any set of finite algebras $\K = \{\A_1,\ldots,\A_m\}$ and $n \in \mathbb{N}$:
\[
|F_\K(n)| \le \prod_{i=1}^m |A_i|^{|A_i|^n}.
\]
\end{lem}

\noindent
Hence for the relatively easy task of checking $\K$-unifiability, it suffices to find (e.g., through exhaustive search) the 
(unique up to isomorphism) smallest subalgebra  $\C$  of the finite free algebra $\F_{\K}(1)$, noting that this is 
$\F_{\K}(0)$ if the language contains constants. 
It then follows by Proposition~\ref{p:unif} that a set of equations $\Si$ is $\K$-unifiable 
if and only if $\Si$ is $\C$-valid, and indeed that there is no smaller algebra with this property. 

For checking $\K$-admissibility, we make use of a known result  for 
finitely generated quasivarieties (see~\cite[Lemma 4.1.10]{Ryb97}), obtained here as 
a corollary of Proposition~\ref{p:adm}:

\begin{cor} \label{c:admfin}
Given a finite set $\K$ of finite $\lang$-algebras with $n = \max\{|A| : \A \in \K\}$:
\begin{enumerate}[\rm(a)]
\item	$\Qg(\F_{\K}(\omega)) = \Qg(\F_{\K}(n))$.
\item	$\Si \imp \f \eq \p$ is $\K$-admissible \ iff \ $\Si \mdl{\F_{\K}({\mathnormal{n}})} \f \eq \p$.
\end{enumerate}
\end{cor}
\begin{proof}
Observe first that each $\A \in \K$ is a homomorphic image of $\F_\K(n)$. That is, define any surjective map from 
the $n$ generators of $\F_\K(n)$ to $A$; this extends to a homomorphism from $\F_\K(n)$ onto $\A$ 
since $\F_\K(n)$ has the universal mapping property for $\K$ over $n$ generators. 
So $\K \subseteq \Vg(\F_\K(n))$ and, since also 
$\F_\K(n) \in \Qg(\F_\K(\omega))$, (a) and (b) follow by Proposition~\ref{p:adm}.
\end{proof}

\noindent
Hence checking $\K$-admissibility of quasiequations is decidable. 
However, even when $\K$ consists of a small number of small algebras, free algebras on a small number of generators 
can be quite large. For example, the free algebra $\F_\alg{D_4}(2)$ (see Example \ref{ex:DeMorganUni}) has $168$ elements. 
We therefore seek smaller algebras or finite sets of smaller algebras that also generate $\Qg(\F_{\K}(\omega))$ as a quasivariety. 
In fact, since $\Qg(\F_\K(\omega))$ is finitely generated, 
we may apply the multiset ordering defined in Section~\ref{s:fingen}, 
and seek a minimal generating set of finite algebras for this quasivariety that is unique up to isomorphism. 
One strategy would therefore be to apply the algorithm \proc{MinGenSet} directly 
to $\F_{\K}(n)$. However, this method is not feasible for large free algebras, since it involves the computationally 
labour-intensive task of building the congruence lattice of $\F_{\K}(n)$. 
Instead, we make auxiliary use of the following immediate corollary of Proposition~\ref{p:adm}:

\begin{cor} \label{c:adm}
Given a class $\K$ of $\lang$-algebras and $\K' \subseteq \op{S}(\F_\K(\omega))$ such that $\K \subseteq \op{H}(\K')$:
\begin{enumerate}[{\rm (a)}]
\item $\Qg(\K') = \Qg(\F_{\K}(\omega))$.
\item $\Si \imp \f \eq \p$ is $\K$-admissible \ iff \ $\Si \mdl{\K'} \f \eq \p$.
\end{enumerate}
\end{cor}

\begin{figure}[t]
\begin{algorithmic}[1]
\Function {AdmAlgs}{$\K$}
   \State \textbf{declare} {$\mathcal{A}, \mathcal{D}$ : set} 
      \State \textbf{declare} $\B, \B'$ : algebra
   \State ${\mathcal{D}} \leftarrow$ \Call{MinGenSet}{$\K$}
   \State ${\mathcal{A}} \leftarrow \emptyset$
   \ForAll{$\A \in {\mathcal{D}}$}
      \State $\B \leftarrow$ \Call{Free}{$\A,{\mathcal{D}}$}
      \State $\B' \leftarrow$ \Call{SubPreHom}{$\A,\B$}
      \While{$\B' \neq \B$} \label{alg:SubPreHom}
         \State $\B \leftarrow \B'$
         \State $\B' \leftarrow$ \Call{SubPreHom}{$\A,\B$}
      \EndWhile
      \State add $\B$ to ${\mathcal{A}}$
   \EndFor
   \State \Return \Call{MinGenSet}{${\mathcal{A}}$}
\EndFunction
\end{algorithmic}
\caption{For a finite set $\K$ of finite algebras, return the minimal generating set of
$\Qg(\F_\K(\omega))$. \label{alg:AdmAlgs}}
\end{figure}

\noindent
We describe an algorithm \proc{AdmAlgs} (see Figure~\ref{alg:AdmAlgs}) 
which takes as input a finite set $\K$ of finite $\lang$-algebras 
and outputs the (unique up to isomorphism) minimal generating set for $\Qg(\F_{\K}(\omega))$. 
First, the procedure \proc{MinGenSet} is applied to $\K$ (which typically would be a 
small set of small algebras) to obtain the set  of $\lang$-algebras $\mathcal{D}$. Then for each $\A \in \mathcal{D}$, 
a procedure \proc{Free($\A,\mathcal{D}$)} is invoked to produce the smallest free algebra 
$\F_{\mathcal{D}}(m)$ having $\A$ as a homomorphic image. (The procedure begins by checking the 
smallest free algebra $\F_\mathcal{D}(0)$ or $\F_\mathcal{D}(1)$, then increases the number of generators one at a time.) 
The algorithm then searches for progressively smaller subalgebras of $\F_{\mathcal{D}}(m)$ which have 
$\A$ as a homomorphic image. More precisely, the procedure $\proc{SubPreHom}(\A,\B)$ searches for 
a proper subalgebra of $\B$ that is a homomorphic image of $\A$, returning $\B$ if no such algebra exists 
(line~\ref{alg:SubPreHom}). This process terminates with a (hopefully reasonably small) algebra which 
is added to a set $\mathcal{A}$. 
Finally, the procedure \proc{MinGenSet} is applied to $\mathcal{A}$. 

\begin{thm}
For a finite set $\K$ of finite $\lang$-algebras, \proc{AdmAlgs}($\K$)
returns the (unique up to isomorphism) 
minimal generating set for the quasivariety $\Qg(\F_\K(\omega))$.
\end{thm}

\begin{exa}\label{ex:no_embedding}
Note that given even just one algebra $\A$, it might not be the case that the smallest subalgebra of $\F_{\A}(\omega)$ having 
$\A$ as a homomorphic image is the smallest algebra that generates $\Qg(\F_{\A}(\omega))$. 
Consider the $4$-element algebra $\alg{P} = \langle \{ a,b,c,d \}, \star \rangle$ where the unary function $\star$ 
and the free algebras $\F_\alg{P}(n)$ are described by the following diagrams:\\

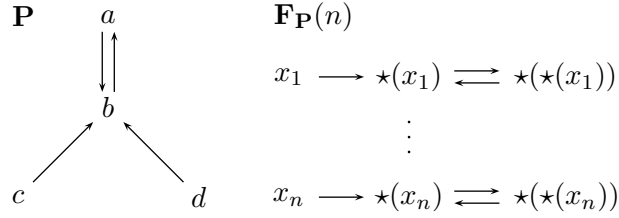
\begin{center}
\pspicture(-1.5,0)(3,3)
\psset{unit=4mm}
\rput[c](-4,7){$a$}
\rput[l](-7.2,7){$\alg{P}$}
\psline[linewidth=0.5pt]{->}(-4.2,6.5)(-4.2,4.5)
\psline[linewidth=0.5pt]{<-}(-3.8,6.5)(-3.8,4.5)
\rput[c](-4,4){$b$}
\psline[linewidth=0.5pt]{->}(-6.5,1.5)(-4.5,3.5)
\rput[c](-7,1){$c$}
\psline[linewidth=0.5pt]{->}(-1.5,1.5)(-3.5,3.5)
\rput[c](-1,1){$d$}
\rput[l](1.5,7){$\F_\alg{P}(n)$}
\rput[c](2,5){$x_1$}
\psline[linewidth=0.5pt]{->}(3,5)(4.5,5)
\rput[c](6,5){$\star(x_1)$}
\psline[linewidth=0.5pt]{->}(7.5,5.2)(9,5.2)
\psline[linewidth=0.5pt]{<-}(7.5,4.8)(9,4.8)
\rput[l](9.5,5){$\star(\star(x_1))$}
\psdots[dotsize=1pt](6,2.5)(6,3)(6,3.5)
\rput[c](2,1){$x_n$}
\psline[linewidth=0.5pt]{->}(3,1)(4.5,1)
\rput[c](6,1){$\star(x_n)$}
\psline[linewidth=0.5pt]{->}(7.5,1.2)(9,1.2)
\psline[linewidth=0.5pt]{<-}(7.5,0.8)(9,0.8)
\rput[l](9.5,1){$\star(\star(x_n))$}
\endpspicture
\end{center}

\noindent 
The smallest subalgebra of $\F_\alg{P}(\omega)$ with $\alg{P}$ as a homomorphic image is the 
$6$-element free algebra $\F_\alg{P}(2)$. However, $\proc{MinGenSet}(\{\alg{P}\})$, performed at the 
beginning of the algorithm \proc{AdmAlgs}, discovers that $\alg{P}$ is a  $\Qg(\alg{P})$-subdirect 
product of (two copies of) the $3$-element free algebra $\F_\alg{P}(1)$. Hence 
$\Qg(\alg{P}) = \Qg(\F_\alg{P}(1))$ and this quasivariety is structurally complete.
\end{exa}

In the previous example, the fact that the algebra $\alg{P}$ is structurally complete is discovered by 
the algorithm \proc{AdmAlgs}. However, structural completeness can also be checked more 
directly using the following result:

\begin{prop} \label{prop:CharSC}
The following are equivalent for any finite set $\K$ of finite $\lang$-algebras:
\begin{enumerate}[{\rm (1)}]
\item $\K$ is structurally complete.
\item $\proc{MinGenSet}(\K) \subseteq \op{IS}(\F_\K(n))$ where $n = \max\{|C| : \C \in \K\}$.
\end{enumerate}
\end{prop}
\begin{proof}
(1) $\Rightarrow$ (2) If $\K$ is structurally complete, then, by Proposition~\ref{p:Bergman} and Corollary~\ref{c:admfin}, 
$\Qg(\K) = \Qg(\F_\K(\omega)) = \Qg(\F_\K(n))$ where $n = \max\{|C| : \C \in \K\}$. 
So $\proc{MinGenSet}(\K) \subseteq \Qg(\F_\K(n))$. But each $\A\in \proc{MinGenSet}(\K)$ 
is $\Qg(\F_\K(n))$-subdirectly irreducible, so by Lemma~\ref{lem:irr_emb}, $\A$ embeds into $\F_\K(n)$. I.e., 
$\proc{MinGenSet}(\K) \subseteq \op{IS}(\F_\K(n))$.\medskip

\noindent (2) $\Rightarrow$ (1) If each  $\A \in \proc{MinGenSet}(\K)$ embeds into $\F_\K(n)$, then 
$\Qg(\F_\K(n)) \subseteq \Qg(\K) = \Qg(\proc{MinGenSet}(\K)) \subseteq \Qg(\F_\K(n))$. So 
$\K$ is structurally complete by Proposition~\ref{p:Bergman}.
\end{proof}

\begin{exa} \label{ex:Pseudocomp}
Proposition~\ref{prop:CharSC} has been used to confirm known structural completeness results 
for the $3$-element positive G{\"o}del algebra $\alg{G^+_3} = \langle \{0,\frac{1}{2},1\}, \min, \max, \to_\lgc{G} \rangle$ 
where  $x \to_\lgc{G} y$ is $y$ if $x > y$, otherwise $1$, and the Stone algebra $\alg{B_1} = 
\langle \{0,\frac{1}{2},1\}, \min, \max, \lnot_{{\lgc{G}}} \rangle$ where $\lnot_\lgc{G} x = x \to_\lgc{G} 0$. 
Moreover, a new structural completeness result has been established for the 
pseudocomplemented distributive lattice $\alg{B_2}$, obtained by adding a top element 
to the $4$-element Boolean lattice, and calculating the pseudocomplement.
\end{exa}

\noindent
Similarly, we can check almost structural completeness using the following result:

\begin{prop}\label{prop:almost}
The following are equivalent for any finite set $\K$ of finite $\lang$-algebras and $\B \in \op{S}(\F_\K(\omega))$:
\begin{enumerate}[{\rm (1)}]
\item $\K$ is almost structurally complete.
\item $\proc{MinGenSet}(\{\A \times \B : \A \in \K\}) \subseteq \op{IS}(\F_\K(n))$ where $n = \max\{|C| : \C \in \K\}$.
\end{enumerate}
\end{prop}
\begin{proof}
(1) $\Rightarrow$ (2) If $\K$ is almost structurally complete, then by Theorem~\ref{thm:ASC}  and Corollary~\ref{c:admfin}, 
$\Qg(\{\A \times \B  :  \A \in \K\}) = \Qg(\F_\K(\omega)) =  \Qg(\F_\K (n))$ where $n = \max\{|{C}| : \C \in \K\}$. In particular, 
$\proc{MinGenSet}(\{\A \times \B : \A \in \K\}) \subseteq \Qg(\F_\K (n))$. 
But each $\C \in  \proc{MinGenSet}(\{\A \times \B : \A \in \K\})$ 
is  $\Qg(\F_\K(n))$-subdirectly irreducible, so by Lemma~\ref{lem:irr_emb}, 
$\C$ embeds into $\F_\K(n)$. I.e.,  $\proc{MinGenSet}(\{\A \times \B : \A \in \K\}) \subseteq \op{IS}(\F_\K(n))$.\medskip

\noindent(2) $\Rightarrow$ (1) If $\proc{MinGenSet}(\{\A \times \B : \A \in \K\}) \subseteq \op{IS}(\F_\K(n))$, 
then $\{\A \times \B : \A \in \K\} \subseteq \Qg(\F_\K(n)) = \Qg(\F_\K(\omega))$. So by Theorem~\ref{thm:ASC}, 
$\K$ is almost structurally complete.
\end{proof}

\begin{exa} \label{ex:SugiMono3}
The Soboci{\'n}ski algebra $\alg{S_3} = \langle \{ -1, 0, 1 \}, \to, \lnot \rangle$ 
with operations 
\[
\begin{array}{r@{\ \ }|@{\ \ }r@{\quad}r@{\quad}rrr@{\ \ }|@{\ \ }rrr}
\to   & -1  & 0   & 1   & \qquad\qquad\qquad & \lnot & \\
\cline{1-4}
\cline{6-7}
-1       & 1   & 1   & 1   & \quad\quad\quad & -1  & 1 \\
0        & -1  & 0   & 1   & \quad\quad\quad & 0   & 0 \\
1        & -1  & -1  & 1   & \quad\quad\quad & 1   & -1 \\
\end{array}
\]
is a weak characteristic matrix for the multiplicative fragment $\lgc{RM_m}$ 
of the substructural logic R-mingle:  $\f$ is 
a theorem of $\lgc{RM_m}$ if and only if $\mdl{\alg{S_3}} \f \eq \f \to \f$. 
Using Proposition~\ref{prop:almost}, we establish that $\alg{S_3}$ is almost structurally 
complete and  that $\Qg(\F_{\alg{S_3}}(\omega)) = \Qg(\alg{S_3} \times \alg{S_2})$ where $\alg{S_2} = 
\langle \{-1,1\}, \to, \lnot \rangle$. Note, moreover, that 
the implicational reduct $\alg{S^{\to}} = \langle \{ -1, 0, 1 \}, \to \rangle$ of $\alg{S_3}$ is structurally complete, 
which can be confirmed directly using Proposition~\ref{prop:CharSC} or discovered automatically by \proc{AdmAlgs}.
\end{exa}

\begin{exa}\label{ex:SugiMono4}
Consider now the algebra $\alg{Z_4} = \langle \{ -2,-1, 1,2 \}, \to, \lnot, 1 \rangle$ (the reduct of a $4$-element 
algebra for R-mingle) with a constant $1$ and operations 
\[
\begin{array}{r@{\ \ }|@{\ \ }r@{\ \ }r@{\quad}r@{\quad}rrr@{\ \ }|@{\ \ }rrr}
\to   & -2  & -1  &  1 & 2    & \qquad\qquad\qquad & \lnot & \\
\cline{1-5}
\cline{7-8}
-2    &  2  &  2  &  2 & 2    & \quad\quad\quad & -2  & 2 \\
-1    & -2  &  1  &  1 & 2    & \quad\quad\quad & -1  & 1 \\
 1    & -2  & -1  &  1 & 2    & \quad\quad\quad & 1   & -1 \\
 2    & -2  & -2  & -2 & 2    & \quad\quad\quad & 2   & -2 \\
\end{array}
\]
 \proc{AdmAlgs} produces a $6$-element algebra (the product of a $2$-element algebra and a 
$3$-element algebra). Remarkably, if we drop $\lnot$ to obtain  $\alg{Z^+_4} = \langle \{ -2,-1, 1,2 \}, \to, 1 \rangle$, 
 \proc{AdmAlgs} produces a $4$-element algebra that is not  isomorphic to  $\alg{Z^+_4}$.
\end{exa}


\begin{exa} \label{ex:DeMorganKleene}
Recall from Example~\ref{ex:DeMorganUni} that the algebras $\alg{D^{\ell}_4}$ and $\alg{D_4}$ generate 
the varieties of De Morgan latices and De Morgan algebras, respectively,  as quasivarieties. 
Both $\alg{D^{\ell}_4}$ and $\alg{D_4}$ are homomorphic images of the corresponding 
free algebras on two generators (with $166$ and $168$ elements, respectively) but not on one generator. 
In the De Morgan lattice case, \proc{AdmAlgs} finds a smallest suitable subalgebra isomorphic to 
$\alg{D^{\ell}_4} \times \alg{2}$, while for De Morgan algebras, 
 the smallest suitable algebra is isomorphic to $\alg{D_4} \times \alg{2}$ with additional top and bottom elements. 
 These results were established ``by hand'' in~\cite{MR12a}; our procedure here confirms them automatically. 
Similar results were also obtained in~\cite{MR12a} for Kleene lattices and Kleene algebras  (subvarieties of De Morgan lattices and 
De Morgan algebras, respectively) which are 
generated as quasivarieties by the $3$-element chains $\alg{C^{\ell}_3} =  \langle \{\top,e,\bot\}, \land, \lor, \lnot \rangle$ 
and $\alg{C_3} =  \langle \{\top,e,\bot\}, \land, \lor, \lnot, \bot, \top \rangle$
where $\lnot$ swaps $\bot$ and $\top$ and leaves $e$ fixed. 
In both cases, the smallest algebra for $\alg{C_3}$  and $\alg{C^{\ell}_3}$,  found automatically by \proc{AdmAlgs} 
is a $4$-element chain.
\end{exa}

\begin{exa}\label{e:twoalgebras}
Consider the $2$-element and $3$-element chains, $\alg{C^e_2} = \langle \{\bot,\top\}, \land, \lor, \lnot, \top \rangle$ and 
 $\alg{C^e_3} = \langle \{\bot, e,�\top\}, \land, \lor, \lnot, e \rangle$ where $\lnot$ swaps $\bot$ and $\top$ and leaves 
 $e$ fixed. Individually, these algebras are structurally complete. However, applying  \proc{AdmAlgs} to 
 $\K = \{\alg{C^e_2}, \alg{C^e_3}\}$, we find that $\K$ is not structurally complete: 
 both $\alg{C^e_2}$ and $\alg{C^e_3}$ are homomorphic images of the $16$-element 
 free algebra $\F_\K(1)$, and the minimal generating set for $\Qg(\F_\K(\omega))$ consists of a single 
 $4$-element algebra.
\end{exa}

We remark finally that except for Example \ref{e:twoalgebras}, all the case studies considered in this section 
(see Table~\ref{table}) involve quasivarieties generated by a single algebra with five elements or fewer, and produce minimal generating sets 
for checking admissibility also containing just one algebra. Indeed in many of the cases considered, the 
quasivariety is shown to be either structurally complete or almost structurally complete. The main obstacle at present to extending 
our experimental results to (larger sets of) larger algebras is the size of the resulting free algebras 
(in the worst case,  size $k^{k^n}$  for $n$ generators and an algebra of size $k$), both for computing these algebras and then for 
generating suitable subalgebras. 
Nevertheless, we hope in future work to obtain procedures capable of treating larger 
examples by using heuristics to constrain the search for appropriate generating algebras.


\begin{table}[tbp]
\begin{center}\small
\begin{tabular}{|@{\ \ }c@{\ \ }|@{\ \ }c@{\ \ }|@{\ \ }l@{\ \ }|@{\ \ }l@{\ \ }|@{\ \ }c@{\ \ }|} \hline
   $\A$                         & $|{A}|$ 	& Quasivariety $\Qg(\A)$   						& Free algebra 		&  $|$\proc{AdmAlgs}($\A$)$|$\\ \hline\hline  
   $\alg{L_3}$      & 3		& algebras for $\lgc{\Luk_3}$ (Ex.~\ref{ex:luk})	 	&  $|{F}_\A(1)|  =   12$ 	&   6  \\ \hline
   $\alg{L_3^\to}$  & 3		& algebras for $\lgc{\Luk^\to_3}$  (Ex.~\ref{ex:luk})			&  $|{F}_\A(2)|  =   40$ 	&   3  \\ \hline
   $\alg{B_1}$                  & 3		& Stone algebras (Ex.~\ref{ex:Pseudocomp})     			&  $|{F}_\A(1)|  =    6$ 	&   3  \\ \hline
   $\alg{C_3}$                  & 3		& Kleene algebras (Ex.~\ref{ex:DeMorganKleene}) 			&  $|{F}_\A(1)|  =    6$ 	&   4  \\ \hline
   $\alg{C^{\ell}_3}$                & 3		& Kleene lattices  (Ex.~\ref{ex:DeMorganKleene})  			&  $|{F}_\A(2)|  =   82$ 	&   4  \\ \hline
   $\alg{S_3}$       & 3		& algebras for $\lgc{RM^{\to\lnot}}$ (Ex.~\ref{ex:SugiMono3})  		&  $|{F}_\A(2)|  =  264$ 	&   6  \\ \hline
   $\alg{S^\to_3}$              & 3		& algebras for $\lgc{RM^\to}$ (Ex.~\ref{ex:SugiMono3})      		&  $|{F}_\A(2)|  =   60$ 	&   3  \\ \hline
   $\alg{G^+_3}$                  & 3		& algebras for $\lgc{G^+_3}$ (Ex.~\ref{ex:Pseudocomp})				&  $|{F}_\A(2)|  =  18$ 	&   3  \\ \hline
   $\alg{D^{\ell}_4}$      & 4		& De Morgan lattices (Exs~\ref{ex:DeMorganUni},\ref{ex:DeMorganKleene})	&  $|{F}_\A(2)|  =  166$ 	&   8  \\ \hline
   $\alg{D_4}$                  & 4		& De Morgan algebras (Exs.~\ref{ex:DeMorganUni},\ref{ex:DeMorganKleene})&  $|{F}_\A(2)|  =  168$ 	&   10 \\ \hline
   $\alg{P}$                    & 4		& $\Qg(\alg{P})$ (Ex.~\ref{ex:no_embedding})        			&  $|{F}_\A(2)|  =    6$ 	&   3  \\ \hline
   $\alg{Z_4}$     & 4		& algebras for $\lgc{RM^{\to\lnot e}}$ (Ex.~\ref{ex:SugiMono4}) &  $|{F}_\A(1)|  =   18$ &   6  \\ \hline
   $\alg{Z_4^+}$          & 4		& algebras for $\lgc{RM^{\to e}}$ (Ex.~\ref{ex:SugiMono4})	&  $|{F}_\A(2)|  =  453$ 	&   4 \\ \hline
   $\alg{B_2}$                  & 5		& $\Qg(\alg{\B_2})$ (Ex.~\ref{ex:Pseudocomp})                    	&  $|{F}_\A(1)|  =    7$ 	&   5  \\ \hline
   $\alg{M_5}$                       & 5		& lattices in $\Qg(\alg{M_5})$  (Ex.~\ref{ex:M5N5})			           	&  $|{F}_\A(3)|  =   28$ 	&   5  \\ \hline
   $\alg{N_5}$                       & 5		& lattices in $\Qg(\alg{N_5})$  (Ex.~\ref{ex:M5N5})           			&  $|{F}_\A(3)|  =   99$ 	&   5  \\ \hline
\end{tabular}
\end{center}\caption{Algebras for checking admissibility}\label{table}
\end{table}


\section{Finite-valued logics}~\label{s:logics}

The preceding characterizations of admissibility can be adapted relatively straightforwardly 
to finite-valued logics. However, in this setting, the designated values of the logic 
-- the ``true truth values'' -- must also be considered, and dealt with appropriately. Here we describe a 
method that given a finite-valued logic $\lgc{L}$, provides another (hopefully small) finite-valued logic $\lgc{L'}$ 
such that validity in $\lgc{L'}$ corresponds to admissibility in $\lgc{L}$. The more general case, where a 
smallest finite set of logics is sought such that validity in all members of the set 
corresponds to admissibility in a logic (or logics), is left as an exercise for the interested reader.

Recall that a {\em finite-valued logic} $\lgc{L} = (\A,D)$ for a language $\lang$ consists of a finite $\lang$-algebra $\A$ and a set of 
{\em designated values} $D \subseteq A$. Given $\Ga \cup \{\f\} \subseteq {\rm Tm}_\lang$, we let $\Ga \der{L} \f$ denote 
that for all homomorphisms $h \colon \alg{Tm_\lang} \to \A$, whenever $h[\Ga] \subseteq D$, also $h(\f) \in D$. A term  
 $\f$ is  {\em $\lgc{L}$-valid} if $\der{L} \f$. 
We  call a logic $\lgc{L_1} = (\A_1,D_1)$ for a language $\lang$ a {\em sublogic} of a logic $\lgc{L_2} = (\A_2,D_2)$ 
if $\A_1$ is a subalgebra of $\A_2$ and also $D_1 = D_2 \cap A_1$. 
 
Consider now a finite-valued logic $\lgc{L} = (\A,D)$ for a language $\lang$ and a finite set of terms 
$\Ga \cup \{\f\} \subseteq {\rm Tm}_\lang$. 
We say that $\Ga$ is {\em $\lgc{L}$-unifiable} if there exists a homomorphism 
 $\si \colon \alg{Tm_\lang} \to \alg{Tm_\lang}$ such that $\der{L} \si(\p)$ for all $\p \in \Ga$ 
 and call $\si$ in this case an  {\em $\lgc{L}$-unifier} of $\Ga$.  
A  rule $\Ga \rl \f$ is said to be {\em $\lgc{L}$-admissible} if every $\lgc{L}$-unifier of $\Ga$ is an $\lgc{L}$-unifier of $\f$.  
Now if we define the finite-valued logic $\lgc{L^*} = (\F_\A(|A|), D^*)$ where $D^* = \{[\f] \in \F_\A(|A|)  : \ \der{L} \f\}$, then 
we easily obtain the following analogue of Theorem~\ref{t:EqAdm}.

\begin{lem}\label{l:finval}
Let $\lgc{L} = (\A,D)$ be a finite-valued logic for a language $\lang$. 
Then $\Ga \rl \f$ is $\lgc{L}$-admissible iff $\Ga \der{L^*} \f$.
\end{lem}

\noindent
The next result may then be understood as an analogue of 
Proposition~\ref{p:adm}.

\begin{prop}
Let $\lgc{L} = (\A,D_A)$ be a finite-valued logic for a language $\lang$, and let $\lgc{L'} = (\B,D_B)$ be a sublogic of 
$\lgc{L^*}$ such that there exists a surjective homomorphism 
$h \colon \B \to \A$ satisfying $h[D_B] \subseteq D_A$.  Then $\Ga \rl \f$ is $\lgc{L}$-admissible iff 
$\Ga \der{L'} \f$.
\end{prop}
\begin{proof}
If $\Ga \rl \f$ is $\lgc{L}$-admissible, then by Lemma~\ref{l:finval}, $\Ga \der{L^*} \f$. Since $\lgc{L'}$ is a sublogic of 
$\lgc{L^*}$, also  $\Ga \der{L'} \f$. Conversely, suppose that $\Ga \der{L'} \f$ and that $\si$ is an $\lgc{L}$-unifier of $\Ga$.  
Notice that if $\der{L} \p$,  then $\der{L^*} \p$ and  $\der{L'} \p$. So $\si$ is also an $\lgc{L^*}$-unifier and 
$\lgc{L'}$-unifier of $\Ga$. But 
$\si(\Ga) \der{L'} \si(\f)$ and therefore $\der{L'} \si(\f)$. 
Now consider any homomorphism $e \colon \alg{Tm}_\lang \to \A$. Since $h$ is a surjective 
homomorphism from $\B$ to $\A$, there exists a homomorphism $k \colon \A \to \B$ 
such that $h \circ k$ is the identity map on $\A$. But $\der{L'} \si(\f)$ and hence $k \circ e \circ \si(\f) \in D_B$. Therefore
$e \circ \si(\f) = h \circ k \circ e \circ \si(\f) \in h[D_B] \subseteq D_A$. So  $\der{L} \si(\f)$.
\end{proof}

\begin{exa}
The $3$-valued \L ukasiewicz logic $\lgc{\Luk_3}$ and Ja{\'s}kowski logic $\lgc{J_3}$ may both be presented using the 
$3$-element 
Wajsberg algebra $\alg{L_3}$ (Example~\ref{ex:luk}) but with $1$ as designated value for $\lgc{\Luk_3}$ 
and  $\frac{1}{2}$ and $1$ as designated values for $\lgc{J_3}$. That is, $\lgc{\Luk_3} = (\alg{L_3}, \{1\})$ and 
$\lgc{J_3} =  (\alg{L_3}, \{\frac{1}{2},1\})$. In this case, there is a smallest subalgebra of $\F_{\alg{L_3}}(\omega)$ 
isomorphic to $\alg{L_3} \times \alg{L_2}$ with a homomorphism that maps $\alg{L_3} \times \alg{L_2}$ onto 
$\alg{L_3}$ and sends the inherited designated values $(1,1)$ to $1$, and $(\frac{1}{2},1)$ to $\frac{1}{2}$. 
We therefore obtain a logic $(\alg{L_3} \times \alg{L_2}, \{(1,1)\})$ corresponding to admissibility in $\lgc{\Luk_3}$, and 
another logic $(\alg{L_3} \times \alg{L_2}, \{(\frac{1}{2},1),(1,1)\})$ corresponding to admissibility in $\lgc{J_3}$.
\end{exa}


\section{Concluding Remarks}\label{s:concluding}

The algorithms described in this paper have been implemented in the tool $\mathsf{TAFA}$, 
 a Delphi XE2 implementation compiled for Windows operating systems available for download 
 online at \url{https://sites.google.com/site/admissibility}. This tool 
has been used in particular to obtain  minimal generating sets and to check 
structural completeness and almost structural completeness for all 
$3$-element groupoids  (i.e.,  $3$-element algebras with a single binary operation) of which there 
are $3330$ up to isomorphism.
The cardinalities of the minimal free algebras required to check admissibility range from
$3$ to $1296$, while the minimal generating algebras have at most $9$ elements. 
$2676$ of the groupoids are structurally complete and a further $254$ are almost structurally complete. 
For further details of $\mathsf{TAFA}$  and these experiments, consult the system description~\cite{CR13}.

Although the algorithms described here produce minimal sets of algebras for testing admissibility 
(and of course have to be run only once for a given finite set of finite algebras), the steps involved 
-- generating a suitable free algebra and its subalgebras and checking homomorphic images 
-- may be computationally unfeasible. For this reason, our case studies have consisted so far  
mostly of algebras with $5$ or fewer elements where the appropriate free algebra is of a 
reasonable size (less than $1000$ elements, say). This bottleneck could perhaps be addressed by 
introducing heuristics such as constructing and checking small subalgebras of free algebras incrementally, 
rather than beginning with the whole free algebra, or calculating only some and not all of the congruences 
in the $\proc{MinGenSet}$ algorithm. We might also hope to obtain faster algorithms for 
cases where the given algebras admit certain properties such as congruence distributivity or permutability. 

We recall finally that one of the motivations for investigating admissibility in arbitrary finite algebras or, similarly, 
finite-valued logics, is to obtain admissible quasiequations or rules that can be used to 
improve proof systems (shortening derivations, constraining proof search, simplifying rules, etc.) 
for these algebras or logics. The work reported here provides a significant step towards this goal. 
The next, equally ambitious, step will be to determine which rules have the potential to be useful 
for a given finite algebra or finite-valued logic and to add these appropriately to corresponding proof systems.


\bibliographystyle{plain}

\begin{thebibliography}{10}

\bibitem{BM10}
F.~Baader and B.~Morawska.
\newblock Unification in the description logic {EL}.
\newblock {\em Log. Methods Comput. Sci.}, 6(3), 2010.

\bibitem{BFS99}
M.~Baaz, C.~G. Ferm{\"u}ller, and G.~Salzer.
\newblock Automated deduction for many-valued logics.
\newblock In {\em Handbook of Automated Reasoning}, volume~II, chapter~20,
  pages 1355--1402. Elsevier, 2001.

\bibitem{BR11b}
S.~Babenyshev and V.~Rybakov.
\newblock Linear temporal logic {LTL}: Basis for admissible rules.
\newblock {\em J. Log. Comput.}, 21(2):157--177, 2011.

\bibitem{BR11a}
S.~Babenyshev and V.~Rybakov.
\newblock Unification in linear temporal logic {LTL}.
\newblock {\em Ann. Pure Appl. Logic}, 162(12):991--1000, 2011.

\bibitem{BRST10}
S.~Babenyshev, V.~Rybakov, R.~A. Schmidt, and D.~Tishkovsky.
\newblock A tableau method for checking rule admissibility in {S}4.
\newblock In {\em Proc. {{M4M}} 2009}, volume 262 of {\em ENTCS}, pages 17--32,
  2010.

\bibitem{Ber91}
C.~Bergman.
\newblock Structural completeness in algebra and logic.
\newblock In {\em Algebraic Logic}, volume~54 of {\em Colloq. Math. Soc.
  {J}{\'{a}}nos {B}olyai}, pages 59--73. North-{H}olland, Amsterdam, 1991.

\bibitem{BS00}
C.~Bergman and G.~Slutzki.
\newblock Complexity of some problems concerning varieties and quasi-varieties
  of algebras.
\newblock {\em SIAM J. Comput.}, 30(2):359--382, 2000.

\bibitem{Bir35}
G.~Birkhoff.
\newblock {On the structure of abstract algebras.}
\newblock {\em Proc. Camb. Philos. Soc.}, 31:433--454, 1935.

\bibitem{Bir40}
G.~Birkhoff.
\newblock {\em Lattice {T}heory}.
\newblock Amer. Math. Soc., New York, 1940.

\bibitem{BS81}
S.~Burris and H.~P. Sankappanavar.
\newblock {\em A Course in Universal Algebra}, volume~78 of {\em Graduate Texts
  in Mathematics}.
\newblock Springer, New York, 1981.

\bibitem{CM201x}
L.~M. Cabrer and G.~Metcalfe.
\newblock Admissibility via natural dualities.
\newblock Manuscript.

\bibitem{XC80}
X.~Caicedo.
\newblock The subdirect decomposition theorem for classes of structures closed
  under direct limits.
\newblock {\em J. Austral. Math. Soc. Ser. A}, 30(2):171--179, 1980/81.

\bibitem{CM09}
P.~Cintula and G.~Metcalfe.
\newblock Structural completeness in fuzzy logics.
\newblock {\em Notre Dame J. Form. Log.}, 50(2):153--183, 2009.

\bibitem{CM10}
P.~Cintula and G.~Metcalfe.
\newblock Admissible rules in the implication-negation fragment of
  intuitionistic logic.
\newblock {\em Ann. Pure Appl. Logic}, 162(10):162--171, 2010.

\bibitem{Dem82}
J.~Demel.
\newblock Fast algorithms for finding a subdirect decomposition and interesting
  congruences of finite algebras.
\newblock {\em Kybernetika}, 18:121--130, 1982.

\bibitem{DeMa79}
N.~Dershowitz and Z.~Manna.
\newblock Proving termination with multiset orderings.
\newblock {\em Commun. ACM}, 22:465--476, 1979.

\bibitem{Ghi99}
S.~Ghilardi.
\newblock Unification in intuitionistic logic.
\newblock {\em J. Symbolic Logic}, 64(2):859--880, 1999.

\bibitem{Ghi00}
S.~Ghilardi.
\newblock Best solving modal equations.
\newblock {\em Ann. Pure Appl. Logic}, 102(3):184--198, 2000.

\bibitem{Ghi02}
S.~Ghilardi.
\newblock A resolution/tableaux algorithm for projective approximations in
  {{IPC}}.
\newblock {\em Log. J. IGPL}, 10(3):227--241, 2002.

\bibitem{Gor98}
V.~A. Gorbunov.
\newblock {\em Algebraic Theory of Quasivarieties}.
\newblock Springer, 1998.

\bibitem{Hah93}
R.~H{\"a}hnle.
\newblock {\em {Automated Deduction in Multiple-Valued Logics}}.
\newblock Oxford Univ. Press, 1993.

\bibitem{Iem01}
R.~Iemhoff.
\newblock On the admissible rules of intuitionistic propositional logic.
\newblock {\em J. Symbolic Logic}, 66(1):281--294, 2001.

\bibitem{IM09b}
R.~Iemhoff and G.~Metcalfe.
\newblock Hypersequent systems for the admissible rules of modal and
  intermediate logics.
\newblock In {\em Proc. {LFCS} 2009}, volume 5407 of {\em LNCS}, pages
  230--245. Springer, 2009.

\bibitem{IM09a}
R.~Iemhoff and G.~Metcalfe.
\newblock Proof theory for admissible rules.
\newblock {\em Ann. Pure Appl. Logic}, 159(1--2):171--186, 2009.

\bibitem{Jer05}
E.~Je{\v r}{{\'a}}bek.
\newblock Admissible rules of modal logics.
\newblock {\em J. Logic Comput.}, 15:411--431, 2005.

\bibitem{Jer09a}
E.~Je{\v r}{\'a}bek.
\newblock Admissible rules of {{\L}}ukasiewicz logic.
\newblock {\em J. Logic Comput.}, 20(2):425--447, 2010.

\bibitem{Jer09b}
E.~Je{\v r}{\'a}bek.
\newblock Bases of admissible rules of {{\L}}ukasiewicz logic.
\newblock {\em J. Logic Comput.}, 20(6):1149--1163, 2010.

\bibitem{Jon61}
B.~J{\'o}nsson.
\newblock Sublattices of a free lattice.
\newblock {\em Canad. J. Math.}, 13:256--264, 1961.

\bibitem{Met12a}
G.~Metcalfe.
\newblock An {{A}}vron rule for fragments of {{R}}-mingle.
\newblock {\em J. Logic Comput.}, to appear.

\bibitem{MOG08}
G.~Metcalfe, N.~Olivetti, and D.~Gabbay.
\newblock {\em Proof Theory for Fuzzy Logics}, volume~36 of {\em Applied
  Logic}.
\newblock Springer, 2008.

\bibitem{MR12a}
G.~Metcalfe and C.~R{\"o}thlisberger.
\newblock Admissibility in {{D}}e {{M}}organ algebras.
\newblock {\em Soft Comput.}, 16(11):1875--1882, 2012.

\bibitem{MR12}
G.~Metcalfe and C.~R{\"o}thlisberger.
\newblock Unifiability and admissibility in finite algebras.
\newblock In {\em Proc. {CiE} 2012}, volume 7318 of {\em LNCS}, pages 485--495.
  Springer, 2012.

\bibitem{Mil79}
G.~L. Miller.
\newblock Graph isomorphism, general remarks.
\newblock {\em J. Comput. System Sci.}, 18(2):128--142, 1979.

\bibitem{CR13}
C.~R{\"o}thlisberger.
\newblock $\mathsf{TAFA}$ -- a tool for admissibility in finite algebras.
\newblock Submitted. Software downloadable from
  \url{https://sites.google.com/site/admissibility}.

\bibitem{Ryb97}
V.~Rybakov.
\newblock {\em Admissibility of Logical Inference Rules}, volume 136 of {\em
  Studies in Logic and the Foundations of Mathematics}.
\newblock Elsevier, Amsterdam, 1997.

\bibitem{Zac93}
R.~Zach.
\newblock Proof theory of finite-valued logics.
\newblock Master's thesis, Technische Universit{\"a}t Wien, 1993.

\end{thebibliography}


\end{document}